\documentclass[10pt,conference]{IEEEtran}
\usepackage{bm}
\usepackage{mathrsfs}
\usepackage{textcomp}
\usepackage{epsfig}
\usepackage{indentfirst}
\usepackage{amsmath}
\usepackage{amssymb}
\usepackage{array}
\usepackage{multirow}
\usepackage{graphicx}
\usepackage{subfigure}

\usepackage{amsthm}
\newtheorem{theorem}{Theorem}%[section]
\newtheorem{lemma}[theorem]{Lemma}

\theoremstyle{definition}

\theoremstyle{remark}

\usepackage{cite}
\begin{document}
\title{
A Stochastic Geometry Analysis of Energy Harvesting in Large Scale Wireless Networks
% A Stochastic Geometry Analysis of Cellular Networks with Energy Harvesting Receivers
% Stochastic Geometry Analysis of Cellular Networks with RF Energy Harvesting Receivers
% Analysis of RF Energy Harvesting in Cellular Networks using Stochastic Geometry
}

\author{
   \IEEEauthorblockN{Ziru Chen\IEEEauthorrefmark{1}, Zhao Chen\IEEEauthorrefmark{1}, Lin X. Cai\IEEEauthorrefmark{1}, Yu Cheng\IEEEauthorrefmark{1}, and Ruoting Gong\IEEEauthorrefmark{2}}
    \IEEEauthorblockA{\IEEEauthorrefmark{1}Department of Electrical and Computer Engineering, Illinois Institute of Technology, Chicago, USA }
    \IEEEauthorblockA{\IEEEauthorrefmark{2}Department of Applied Mathematics, Illinois Institute of Technology, Chicago, USA}
 %\IEEEauthorblockA{Emails: \IEEEauthorrefmark{1}$skhairy$@hawk.iit.edu, \{$lincai, cheng$\}@iit.edu,   \IEEEauthorrefmark{2}$zhan2$@uh.edu, \IEEEauthorrefmark{3}$hshan$@zju.edu.cn}
\IEEEauthorblockA{zchen71@hawk.iit.edu, \{zchen84, lincai, cheng, rgong2\}@iit.edu}
}

\maketitle

\begin{abstract}
In this paper,  the theoretical sustainable capacity of wireless networks  with radio frequency (RF) energy harvesting is analytically studied. Specifically, we consider a large scale wireless network where base stations (BSs) and low power wireless devices are deployed by homogeneous Poisson point process (PPP) with different spatial densities. Wireless devices exploit the downlink transmissions from the BSs for either information delivery or energy harvesting. Generally, a BS schedules downlink transmission to wireless devices. The scheduled device receives the data information while other devices harvest energy from the downlink signals. The data information can be successfully received by the scheduled device only if the device has sufficient energy for data processing, i.e., the harvested energy is larger than a threshold. Given the densities of BSs and users, we apply stochastic geometry to analyze the expected number of users per cell and the successful information delivery probability of a wireless device, based on which the total network throughput can be derived. It is shown that the maximum network throughput per cell can be achieved under the optimal density of BSs. Extensive simulations validate the analysis.

\end{abstract}


\section{Introduction}
Recently, new enhancements for cellular Internet of things (IoT) communication are proposed in the 3GPP releases~\cite{gozalvez2016new}. As the era of IoT approaches, wireless charging technology has been proposed as a promising solution to supply power to a massive number of wireless IoT devices~\cite{Cai_JSAC14,cai2011dimensioning}. By exploiting wireless signals in the surrounding environment coming from TV towers~\cite{vyas2013wehp} and 
%wireless networks including Wi-Fi~\cite{ermeey2016indoor} and 
cellular base stations (BSs)~\cite{chen2017optimal,chen2017sustainable,zhao_TVT14}, wireless charging enables low-power wireless IoT devices to scavenge ambient radio frequency (RF) energy and convert it into DC power to keep themselves alive without replacing or recharging their batteries.

Stochastic geometry~\cite{haenggi2009stochastic,andrews2011tractable,Cai_TSG13} has been widely adopted to characterize the random deployment of BSs and wireless users for network performance analysis, which can be employed to quantify the co-channel interference~\cite{zhang2014stochastic} and also be incorporated with random channel effects such as fading and shadowing~\cite{keeler2013sinr}.
More recently, stochastic geometry approaches have been applied to analyze the capacity of RF energy harvesting in wireless networks~\cite{sakr2015analysis,ding2014wireless,flint2015performance,khan2016millimeter}.
%In~\cite{huang2014enabling}, the deployment of power beacon and data BSs was modeled to analyze the outage probability of uplink transmissions.
In~\cite{sakr2015analysis}, a k-tier heterogeneous network was modeled by Poisson point process (PPP) and users harvest energy from ambient RF signals for uplink transmissions. In~\cite{ding2014wireless}, the authors considered wireless information and power transfer in cooperative networks with spatially random deployed relays.
By employing a general repulsive point process~\cite{flint2015performance}, point-to-point uplink transmission powered by ambient RF energy was studied and power and transmission outage probability was derived.
Besides, energy harvesting over millimeter-wave band with a certain  beamwidth was investigated in~\cite{khan2016millimeter}.
These aforementioned works studied energy harvesting enabled point-to-point transmission performance over a single link. To the best of our knowledge, there is no analytical study on the energy harvesting enabled network performance in a multi-point to multi-point wireless network environment with multi-user scheduling.

Generally, a low-power IoT device can harvest RF energy when it is not scheduled for data communications. When the user density increases, a user may wait for a longer time for data communications, yet in the meantime it may harvest more RF energy to fulfill the data communication requirement. Research works on link layer performance study mainly focused on the impacts of the density of BSs but ignored the density of users on the downlink transmission performance, which also plays a critical role in energy harvesting and transmission scheduling. This motivates us to investigate the network performance under different densities of BSs and users and answer the following questions: how the throughput performance is impacted under different densities of BSs and low-power users, how can we achieve the best throughput performance by determining the density of BSs, and under what conditions the users can operate sustainably.
Note that in multiuser network, given the density of users, although each user can harvest more energy as the density of BSs increases,
the intra-cell interference becomes more severe and the low-power user may not have enough time to harvest sufficient RF energy since there will be a fewer number of users in a cell. Meanwhile, given the density of BSs, there exists a minimal density of the users that ensures energy sustainable operation of users, i.e., a user can always harvest sufficient RF energy for data communication when it is scheduled, such that the maximum network throughput can be achieved.

The main contribution of the paper is four-fold. First, we apply stochastic geometry to analyze the  signal to interference and noise ratio (SINR) of a single user. Second, we further derive the probability that a user has enough energy for information processing when it is scheduled for data reception, which is denoted as the successful information delivery probability. Third, based on the derived information delivery probability, the per cell network throughput and total network throughput can be obtained. We show that for a given density of users, the maximum per cell throughput can be achieved under the optimal density of BSs. Finally, extensive simulations are conducted to validate the analysis.

The remainder of this paper is organized as follows. The system model is presented in Section~\ref{sec:sysmod}. The analysis of the successful information delivery probability and the network throughput are provided in Section III and Section IV, respectively. Then, numerical results are given in Section V, followed by concluding remarks in Section IV.

% Specifically, given the density of BSs, the transmission probability and the network throughput will increase as the density of MUs grows, which will saturate as some point and make the network system fully sustainable.

% story and contribution in this paper

\section{System model} \label{sec:sysmod}

We consider a high-density cellular network, where the deployments of BSs and wireless powered IoT devices follow homogeneous Poisson point process (PPP) with different densities, as shown in Fig.~\ref{pic:by:demo}. Specifically, on the Euclidean plane $\mathbb{R}^2$, the locations of BSs are modeled by a homogeneous PPP $\Phi_{B}=\{b_{i}:i=1,2,\ldots\}$ of spatial density $\lambda_{B}$, while the deployment of users are modeled by a homogeneous PPP $\Phi_{U}=\{u_{i}:i=1,2,\ldots\}$ of spatial density $\lambda_{U}$.
In the network, each wireless device or wireless user will associate with the nearest BS and the coverage area of each BS comprises a Voronoi cell.

In each Voronoi cell, the BS $b_i \in \Phi_{B}$ schedules time slotted downlink transmissions to its associated users, i.e., each user takes turns to be served in each slot in a round robin fashion. Thus, for a cell with $N$ users, each user will transmit once for every $N$ slots and only one user is scheduled for downlink communications in one cell at any time slot. Each cell may have different numbers of users, however, without loss of generality, we consider that all BSs schedule their associated users one round after another, i.e., BSs always transmit to its associate users. %, while other unscheduled users can take this opportunity for energy harvesting from concurrent downlink transmissions from all surrounding BSs.
Since a wireless user can either harvest energy from all surrounding BSs or receive downlink data from the associated BS, a ``harvest then receive'' strategy is adopted by wireless users. That is, each user can harvest energy and store the harvested energy in a capacitor when it is not scheduled for data communications; and use the stored energy for data information reception when it is scheduled for downlink communication by the serving BS~\cite{salter2009rf}~\footnote{With the use of capacitor, a wireless user depletes its stored energy for one data communication}.
It is worth noting that information delivery will proceed successfully only when the user has harvested sufficient RF energy when the user is scheduled for data communications. If the energy is not sufficient for receiving the data, the wireless user will keep  harvesting energy from the downlink transmissions yet no information will be delivered. Denote the downlink transmission power of BSs as $P_S$, and the energy consumption of data reception in each slot as $E_{th}$.

% For each BS $b_i \in \Phi_{B}$, we consider a TDMA-based downlink transmission for the serving users in its cell with  an identical transmission power $P_S$. If there are multiple users associated by the BS in the cell, only one user is assumed to be served at any given time slot. The scheduler choose the user with round-robin policy. In the first scenario, all of the BSs will transmit signals continually even there are not any user to be served in the cell.

%For each wireless powered user $u_i \in \Phi_U$, it will expend the RF energy harvested from downlink wireless signals emitted by surrounding BSs to receive and decode data signals from its associated BS.

Without loss of generality, the received RF power of a typical user %located at $(0,0)\in\mathbb{R}^2$  %cl: why (0,0)? 
at any time slot is given by
\begin{align}\label{eq.PH}
&P_{\mathrm{H}} = \sum_{b_i\in\Phi_{B}} a{P_S h_{i}r_{i}^{-\alpha}},
\end{align}
where $a \in (0,1)$ is the RF-to-DC energy conversion efficiency, $h_i \sim \exp(1)$ is the i.i.d. Rayleigh channel fading,  $r_i$ is  the distance between a typical user and any BS $b_i \in \Phi_B$, and $\alpha$ is the path loss exponent. % which usually takes a value larger than $2$.

For a typical user scheduled for downlink communications,  the received signal power from its associated BS $b_1$
is interfered by the wireless signals from other concurrent transmissions from other BSs.
As a result, the received SINR of a typical user can be written as
\begin{align}
&\mathrm{SINR} = \frac{P_S h_1 r_1^{-\alpha}}{\sigma^2+\sum_{b_i \in \Phi_{B}/b_{1}} P_S h_{i} r_{i}^{-\alpha}} = \frac{P_S h_1 r_1^{-\alpha}}{\sigma^2+I_{1}},
\end{align}
where the aggregate interference of a typical user is $I_1 = \sum_{b_i \in \Phi_{B}/b_{1}} P_S h_{i} r_{i}^{-\alpha}$. Here, $r_1$ denotes the distance of a typical user to its nearest BS $b_1$. According to \cite{andrews2011tractable}, the distribution of distance $r_1$ is given by:
\begin{align}
f_{R_1}(r_1) = 2\pi \lambda_B r_1 e^{-\lambda_B\pi r_1^2}.
\end{align}

In the following sections, we analyze the downlink throughput performance of a wireless network with energy harvesting under different deployment densities of BSs and users. %, $\lambda_B$ and $\lambda_U$. % where the transmission probability of each user and the average throughput performance will be extensively investigated both analytically and numerically.
\begin{figure}[t]
\centering
\includegraphics[width=7cm]{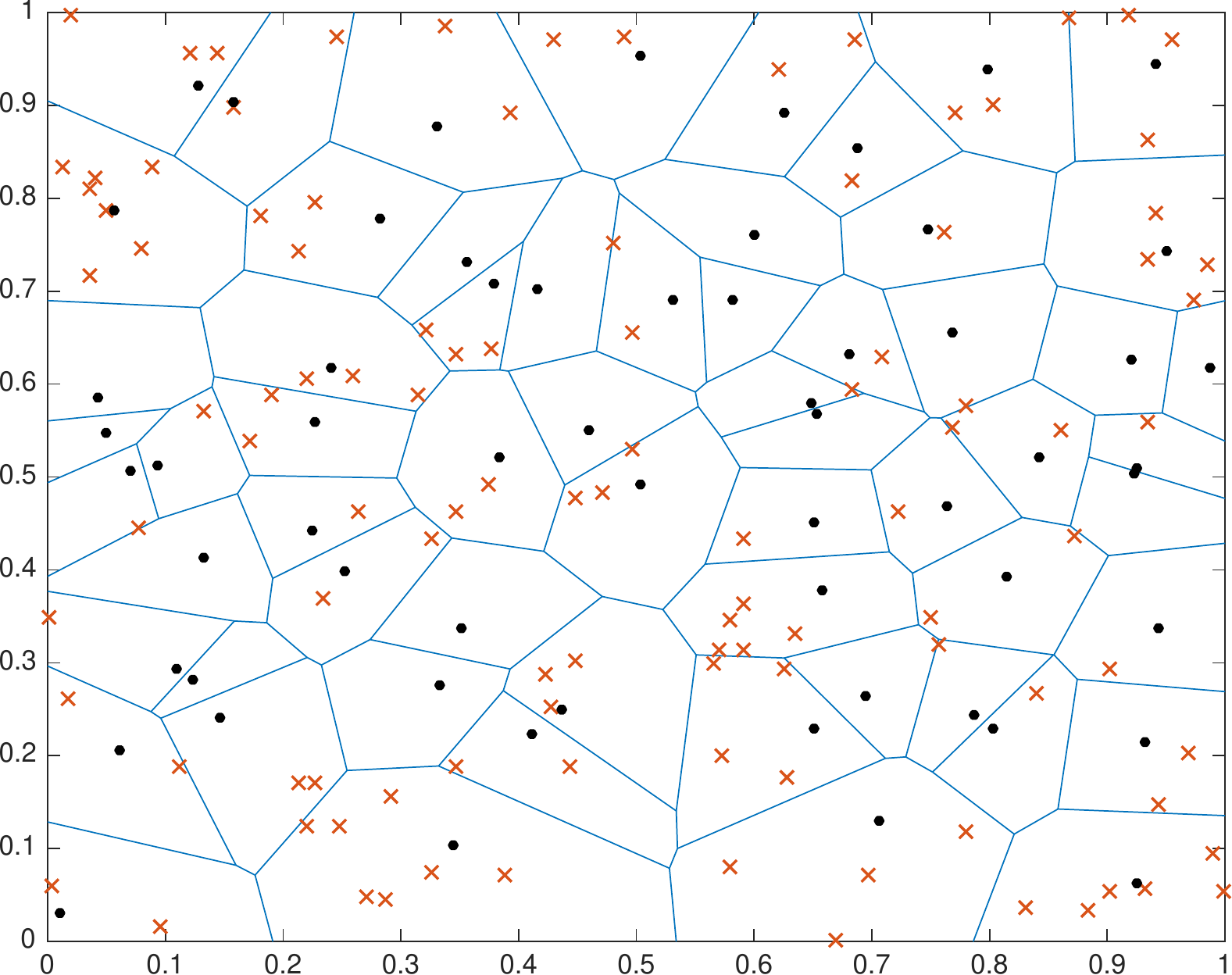}
\caption{Illustration of a Poisson distributed cellular network with each mobile associated with the nearest BS. The cell boundaries are shown and form a Voronoi tessellation.}
\label{pic:by:demo}
\end{figure}

\section{Successful Information Delivery Probability}
In this section, we study the successful information delivery probability, i.e., the probability that a typical user successfully receives data from its associated BS. For a user to successfully receive data, the user needs to harvest sufficient energy for data processing when it is scheduled. Thus, we first analyze the energy harvesting performance of a wireless user.

% is scheduled for data communication ready to receive with sufficient harvested energy when it is scheduled by the serving BS, is provided for downlink transmission.
%Firstly, the probability of typical user ready to transmit after given number of TDMA rounds will be derived, then the transmission probability can be obtained.

\subsection{Energy Harvesting Analysis}

Given that there are $N+1$ users in a Voronoi cell including a typical user and $N$ other users, and $r_1$, the distance between the typical user and its associated BS, according to \eqref{eq.PH},
% The value of energy harvested in one time slot has been mentioned in \eqref{eq.PH}. We denote $N$ as the number of other users in the same Voronoi cell with the typical user. Since we will analysis the typical user only when it has been selected, we have to analyze whether the typical user is ready to receive after given rounds.
we can obtain the probability that a typical user is ready to receive data after $k$ rounds of transmission schedules, i.e., the total energy harvested by the typical user exceeds the threshold $E_{\mathrm{th}}$ after $k$ rounds, as in the following lemma.

\begin{lemma}\label{lm.EK}
For a typical user, the probability that the total energy harvested after $k$ rounds of transmissions exceeds the threshold $E_{\mathrm{th}}$ conditioned on the number of other users in the cell and the nearest distance $r_1$ can be approximated by
\begin{align}
&\mathbb{P}\left(\sum_{j=1}^{k(n+1)-1} P_{\mathrm{H},j} \geq E_{\mathrm{th}}  \bigg| N=n, R=r_1 \right) \nonumber\\
\approx &\sum_{j=0}^{k-1} \frac{\exp(-\Theta(k,n,r_1))\Theta(k,n,r_1)^j}{j!}, \nonumber
\end{align}
where $P_{\mathrm{H},j}$ is the energy harvested in a slot, $j \in \{1,\ldots,k\}$ and $\Theta(k,n, r_1) = \frac{E_{\mathrm{th}}r_1^\alpha}{aP_S}-\frac{2\pi [k(n+1)-1]\lambda_Br_1^{2}}{\alpha-2}$.

\end{lemma}

\begin{proof}
Please refer to Appendix \ref{app.EK}.
\end{proof}

With Lemma 1, if a typical user requires $K$ rounds of transmissions to satisfy the energy threshold $E_{th}$, the probability mass function (PMF) of $K$ can be derived as follows,
\begin{align}
&\mathbb{P}(K =k |N=n,R=r_1) \nonumber\\
= &\mathbb{P}\!\left(\!\sum_{j=1}^{(k-1)(n+1)-1}\!P_{\mathrm{H},j}<E_{\mathrm{th}}\leq \sum_{j=1}^{k(n+1)-1}\!P_{\mathrm{H},j}  \bigg|\!N=\!n, R=r_1\!\right)  \nonumber \\
= &\mathbb{P}\left(\sum_{j=1}^{k(n+1)-1} P_{\mathrm{H},j} \geq E_{\mathrm{th}}  \bigg| N=n,R=r_1  \right)\nonumber\\
 & - \mathbb{P}\left(\sum_{j=1}^{(k-1)(n+1)-1} P_{\mathrm{H},j} \geq E_{\mathrm{th}}  \bigg| N=n,R=r_1  \right)\!. \label{eq.K_r_1}
\end{align}

\subsection{Information Delivery Analysis}
A typical user can successfully receive data only when the user is scheduled for information delivery and it has harvested sufficient RF energy. 
In order to analyze the information delivery probability, we first condition on a given distance $r_1$, the distance between a typical user and its associated BS. 
\begin{align}\label{eq.pt}
P(\mathrm{Tr}) = \int_{0}^{+\infty} P(\mathrm{Tr}|R=r_1) f_{R_1}(r_1)dr_1,
\end{align}
where $P(\mathrm{Tr}|R=r_1)$ denotes the information delivery probability conditioned on $r_1$. 
Basically, the conditional probability will increase as the distance $r_1$ becomes smaller because more energy can be harvested at a smaller distance $r_1$.
% Following the assumption that the typical user will deplete all the harvested energy stored for each successful transmission,
%Thus, the information delivery probability $P(\mathrm{Tr})$ can be derived from the conditional probability
$P(\mathrm{Tr}|R=r_1)$ can be further derived from $P({\mathrm{Tr}}|N=n,R=r_1)$ as
\begin{align}\label{eq.Trans_r1}
P(\mathrm{Tr}|R=r_1) =\! \sum_{n=0}^{\infty}\! {P}({\mathrm{Tr}}|N=n,R=r_1) \mathbb{P}(N=n|R=r_1).
\end{align}
%In order to derive $P(\mathrm{Tr}|R=r_1)$, we need to further calculate the conditional probability ${P}({\mathrm{Tr}}|N=n,R=r_1)$ which denotes the information delivery probability conditioned on distance $r1$ and the number of other users in the typical cell $n$. 
Conditioning on $K=k$, the minimum number of transmission rounds for a typical user to harvest sufficient energy, we have 
\begin{align}
& {P}({\mathrm{Tr}}|N = n,R=r_1) \nonumber \\
= & \sum_{k=1}^\infty {P}({\mathrm{Tr}}|N=n,K=k,R=r_1) \mathbb{P}(K=k |N=n, R=r_1).\label{eq.TR_n_r1}
\end{align}

%the minimum number of transmission rounds needed to harvest sufficient energy for a typical user
%mo
Since the information delivery probability is defined as the probability that a typical user has harvested enough energy to succcessfully receive data when it has been selected, if the typical user needs to harvest energy for $k$ rounds before
it has enough energy to receive data at the scheduled slot in the
last round, it will has a
successful information delivery once for every $K = k$ rounds
of scheduled transmissions. When $k$ is given, the successful
information delivery probability is independent of the number
of other users in the typical cell and the distance $r_1$, the
conditional information delivery probability is
% A typical user needs to harvest energy for $k$ rounds before it can successfully receive data at the scheduled slot in the last round. The user is scheduled for transmission $k$ times but only succeed once. Thus, a typical user will have a successful information delivery once for every $K=k$ rounds of scheduled transmissions. When $k$ is given, the successful information delivery probability is independent of the number of other users in the typical cell and the distance $r_1$, the conditional information delivery probability is
\begin{align}\label{eq.energy_round}
&{\mathbb{P}}({\mathrm{Tr}}|N=n,K=k,R=r_1)=\frac{1}{k}.
\end{align}

Therefore, by combining \eqref{eq.TR_n_r1} and \eqref{eq.energy_round}, we have
\begin{align}
& {P}({\mathrm{Tr}}|N = n,R=r_1) =  \sum_{k=1}^\infty \frac{1}{k}\mathbb{P}(K=k |N=n, R=r_1).
\end{align}

%Considering the fact that the users are
%%cl: did you discuss this work in related work paragraph and highlight your new contribution compared with this work. It NEEDS to be done.
%distributed with PPP of density $\lambda_U$,
To obtain $P(\mathrm{Tr}|R=r_1)$, we need to derive $\mathbb{P}(N=n|R=r_1)$.  Denote $X$ as the area of the cell where the typical user is located, we have
\begin{align}
&\mathbb{P}(N=n|R=r_1)  \nonumber\\
= &\int_{0}^{+\infty} \mathbb{P}(N=n|X=x,R=r_1) f_{X|R_1}(x)dx. \hspace{-0.5em} \nonumber \label{eq.N_r1}
\end{align}
where $f_{X|R_1}(x)$ is the PDF of $X$ conditioned on $r_1$.
For a given $X$, $N$ and $R$ are conditionally independent, and we have
\begin{align}
\mathbb{P}(N\!=\!n|X\!=\!x,R=r_1)\!=\!\mathbb{P}(N\!=\!n|X\!=\!x) = \frac{(\lambda_U x)^n}{n!} e^{-\lambda_U x},\nonumber
\end{align}
 and$f_{X|R_1}(x)$ can be obtained as
\begin{align}
f_{X|R_1}(x) = \frac{f_{R_1|X = x}(r_1)f_{X}(x)}{f_{R_1}(r_1)},
\end{align}
with the conditional PDF $f_{R_1|X = x}(r_1)$ given in the following lemma.
%cl: add R_1=r_1 in some other equations where you use r_1 directly?

\begin{lemma}\label{lm.X_r1}
For a typical user, the conditional PDF of the distance $r_1$ conditioned on the cell area $x$ that it falls in can be approximated by
\begin{align}
f_{R_1|X=x}(r_1) \approx  \frac{6.029 r_1}{x}e^{-3.891 \left(\frac{r_1}{x}\right)^{2.7} }
\end{align}
\end{lemma}

\begin{proof}
Please refer to Appendix \ref{app.X_r1}.
\end{proof}

\section{Network Throughput Analysis}
Based on the derived successful information delivery probability, we can analyze the network throughput performance of a wireless network with energy harvesting. The expected total network throughput is the product of the density of BSs with associated users, $\lambda_B^{\prime}$, and the average per cell throughput, $T_{\mathrm{avg}}$,
\begin{align}\label{eq.throughput}
T = \lambda_B^{\prime} T_{\mathrm{avg}},
\end{align}
where $\lambda_B^{\prime}=\lambda_B(1-(1+3.5^{-1}\lambda_u/\lambda_B)^{-3.5})$\cite{yu2013downlink} is the density of non-empty cells. To obtain total network throughput, we need to analyze the average per cell throughput.

The average per cell throughput is the achieved downlink transmission rate of a user, given that the user can successfully receive data with probability $\mathbb{P}(\mathrm{Tr})$. Thus, we have
\begin{align}\label{averageT}
T_{\mathrm{avg}}
=& \mathbb{E}[\log(1+\mathrm{SINR})\cdot \mathbb{P}{(\mathrm{Tr})}] %\\
%=&  \mathbb{E}_{r_1}\left[\log(1+\mathrm{SINR})\cdot \mathbb{P}{(\mathrm{Tr})} |R=r_1 \right] \\
% =  &\mathbb{E}_{r_1}\left[\mathbb{E}_{h} \left[\log(1+\mathrm{SINR})|R=r_1 \right]\cdot \mathbb{P}{(\mathrm{Tr}|R=r_1)}\right]. \label{eq.T_avg_r1}
\end{align}
%cl: re-arrange the equations to make it logically smooth. 13-15, 14-16 are smooth, but currently you mix them up.

For notation simplicity, define $C = \log(1+\mathrm{SINR})$. Condition on $r_1$, (\ref{averageT}) can be re-written as 
%\begin{align}\label{averageT1}
%T_{\mathrm{avg}} &= C \cdot \mathbb{P}{(\mathrm{Tr})}                         %\\
%= \mathbb{E}_{r_1}\left[C \cdot \mathbb{P}{(\mathrm{Tr})} |R=r_1 \right]   \\
%&=  \mathbb{E}_{r_1}\left[\mathbb{E} \left[ C|R=r_1 \right]\cdot \mathbb{P}{(\mathrm%{Tr}|R=r_1)}\right]. \label{eq.T_avg_r1}
%\end{align}
%cl: why E_h? Should we remove _h. E[C|r1]

%Chen
\begin{align}\label{averageT1}
T_{\mathrm{avg}} &= \mathbb{E}[C \cdot \mathbb{P}{(\mathrm{Tr})}]\nonumber\\
&= \int_{0}^{+\infty}\mathbb{E}\left[C \cdot \mathbb{P}{(\mathrm{Tr})} |R=r_1 \right] \cdot f_{R_1}(r_1)dr_1  \\
&= \int_{0}^{+\infty}\mathbb{E} \left[ C|R=r_1 \right]\cdot \mathbb{P}{(\mathrm{Tr}|R=r_1)} \cdot f_{R_1}(r_1)dr_1. \label{eq.T_avg_r1}
\end{align}

We can derive the conditional probability of $C$ in the following lemma.
%cl: you should define C in an earlier place before 13, so explains things more smoothly. You do not list equations but needs to explain them clearly. Please check what I have added after some equations, and try to describe them clearly in all places.
\begin{lemma}\label{lm.SINR_threshold}
For a typical user, the probability of link capacity $C$ to be larger than a value $t$ conditioned on the distance $r_1$ is given by
\begin{align}
\mathbb{P} (C \geq t | R=r_1) = e^{-\frac{ (2^t-1) r_1^{\alpha}\sigma^2}{P_S} - \pi\lambda_{B}r_1^2 \rho(2^t-1) },
\end{align}
where $\rho(x) = x^{\frac{2}{\alpha}} \int^{+\infty}_{x^{-\frac{2}{\alpha}}} \frac{du}{1+u^{\frac{\alpha}{2}}}$.
\end{lemma}
\begin{proof}
Please refer to Appendix \ref{app.SINR_threshold}.
\end{proof}
Based on Lemma \ref{lm.SINR_threshold}, 
\begin{align}\label{eq.cell_throughput}
\mathbb{E} \left[C|R=r_1 \right] = \int_{0}^{+\infty} \mathbb{P} (C \geq t | R=r_1) dt.  %I remove _h in \mathbb{E}_{h}
\end{align}	
Substituting \eqref{eq.Trans_r1} and \eqref{eq.cell_throughput} in \eqref{eq.T_avg_r1}, we can obtain the average per-cell throughput conditioned on $r_1$, from which the network throughput $T_{\mathrm{avg}}$ can be derived by integrating over $r_1$. %unconditioning on $R_1=r_1$.
%cl: please fill in the correct equation reference of (16) in the above sentence,

%%%%%%%%%%%%%%%%%%%%%%%%%%%%%%%%%%%%%%%%%%%%%%%%%%

%\begin{figure*}[!htb]
%\captionsetup{justification=centering}
%\begin{minipage}{0.3\textwidth}
%  \includegraphics[width=\linewidth]{Trans_pro_final.pdf}
%\caption{Comparison of analytical and simulation results for the transmission probability under different density of BSs $\lambda_B$. }
%\label{fig.lu_trans}
% \end{minipage}
% \begin{minipage}{0.3\textwidth}
%   \includegraphics[width=\linewidth]{E_n.pdf}
%\caption{Comparison of analytical and simulation results for the average number of users in one cell under different density of BSs $\lambda_B$. }
%\label{fig.lu_cap_2}
%  \end{minipage}
%  \begin{minipage}{0.3\textwidth}
%  \includegraphics[width=\linewidth]{th_all_final.pdf}
%\caption{Comparison of analytical and simulation results for the average throughput under different density of BSs $\lambda_B$.}
%\label{fig.lu_cap}
%  \end{minipage}
%  \end{figure*}

%%%%%%%%%%%%%%%%%%%%%%%%%%%%%%%%%%%%%%%%%%%%%%%%%%%%

\begin{figure}[t]
\centering
\includegraphics[width=0.44\textwidth]{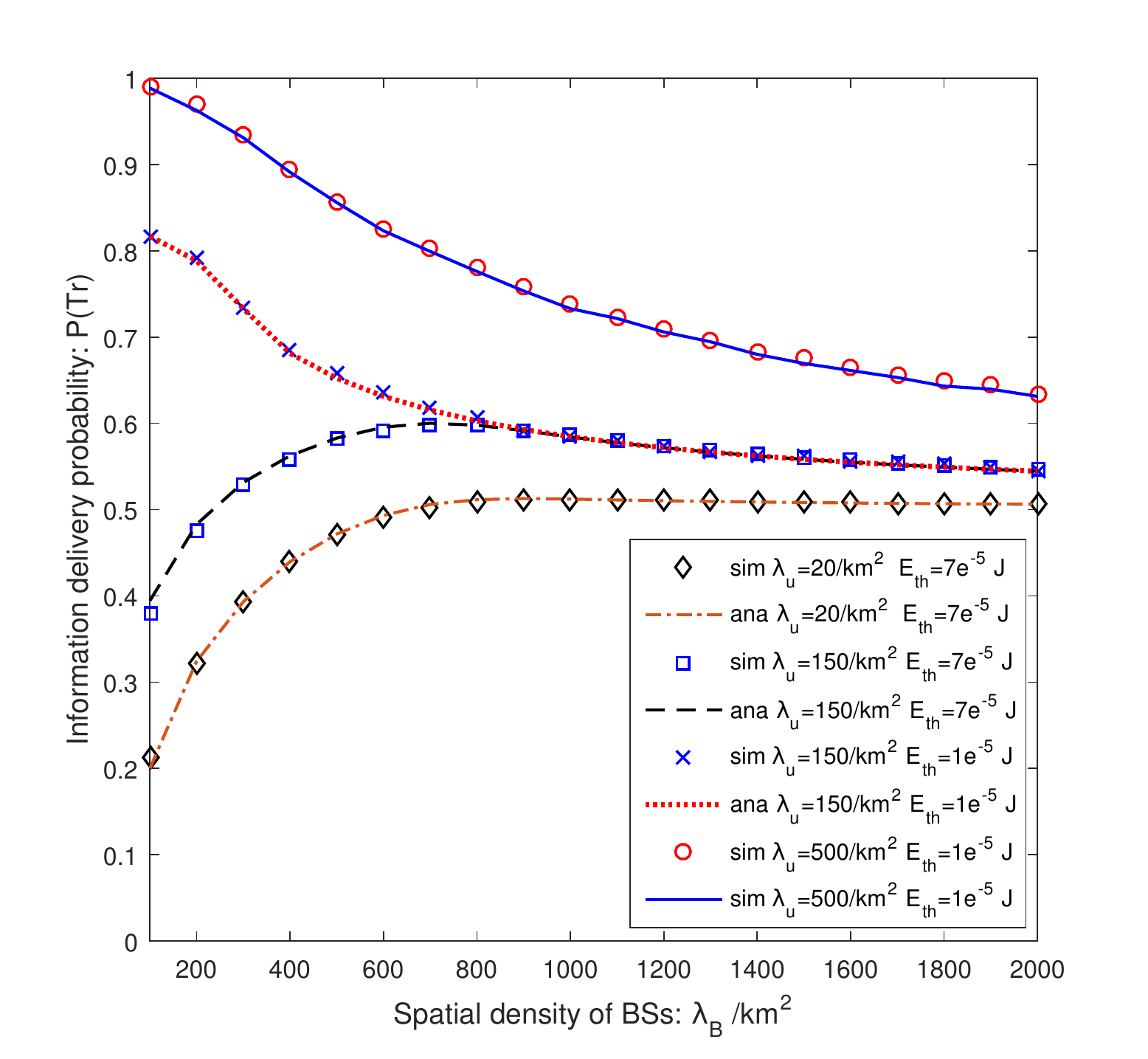}
\caption{Successful information delivery probability under different densities of BSs $\lambda_B$. }
\label{fig.lu_trans}
\end{figure}
%cl: delete all 'Comparison of analytical and simulation results' in the figure title. title needs to be concise and clear

\begin{figure}[t]
\centering
\includegraphics[width=0.44\textwidth]{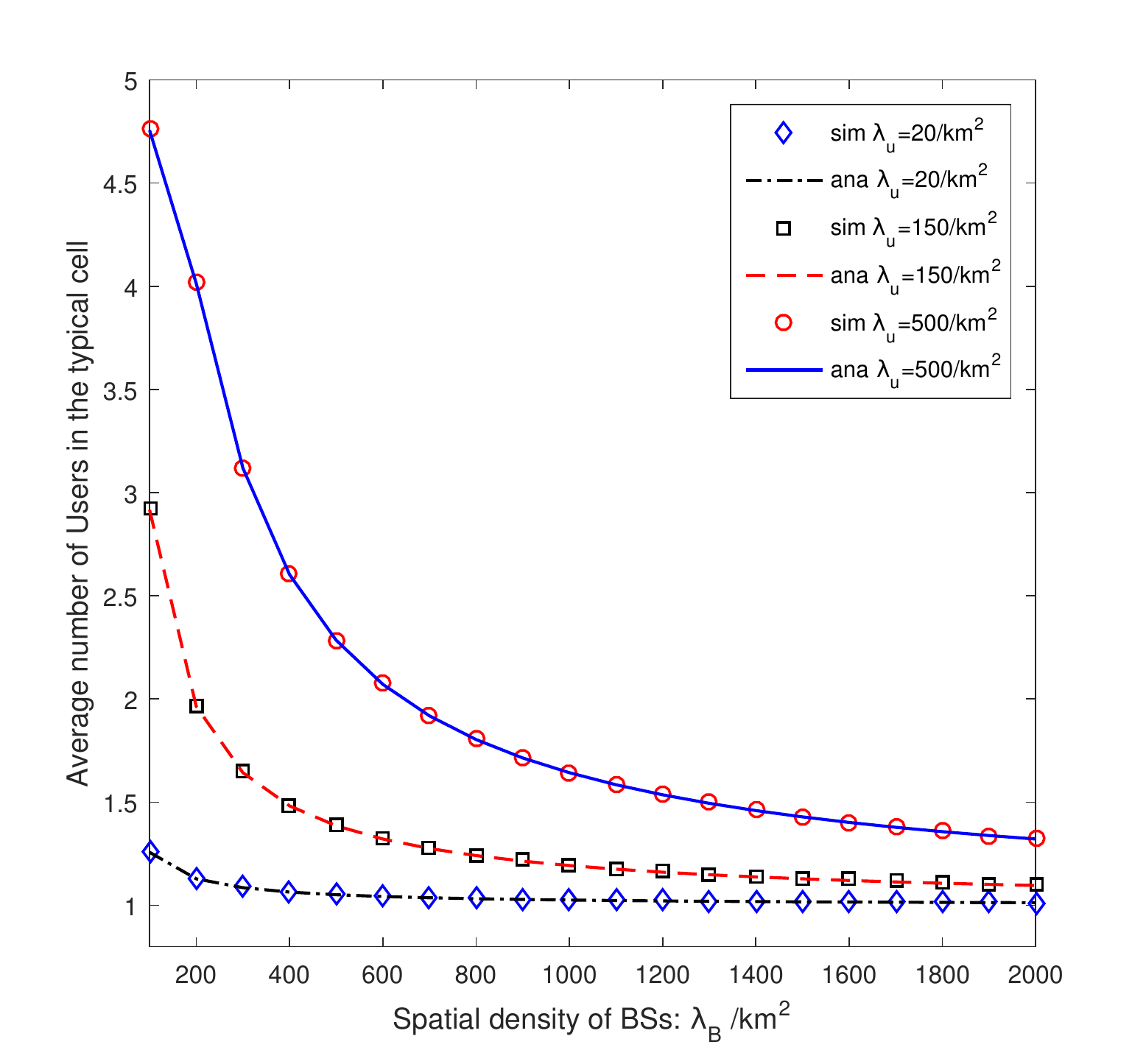}
\caption{The average number of users in one cell under different densities of BSs $\lambda_B$. }
\label{fig.lu_cap_2}
\end{figure}

\begin{figure}[t]
\centering
\includegraphics[width=0.44\textwidth]{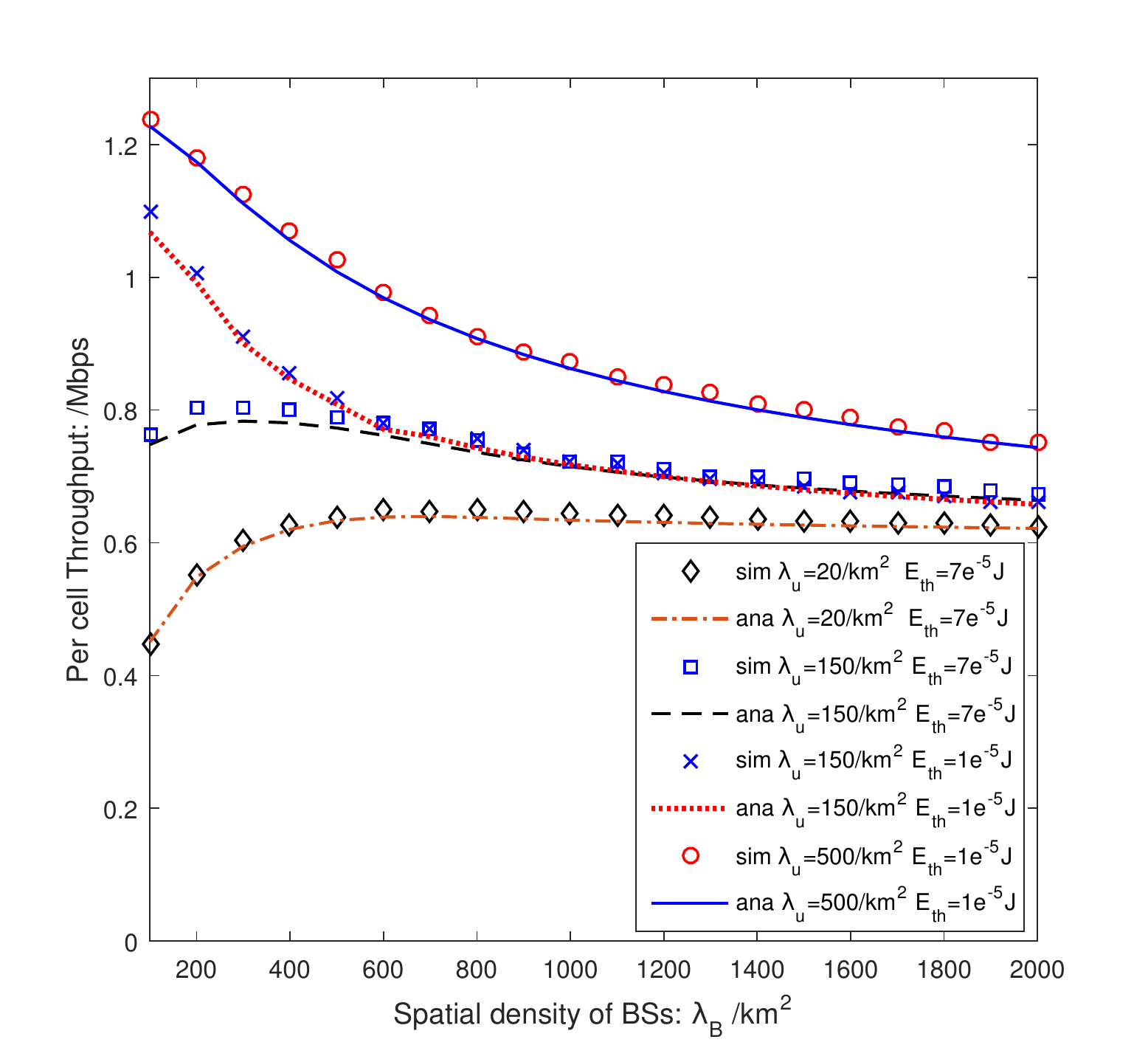}
\caption{The average throughput under different densities of BSs $\lambda_B$.}
\label{fig.lu_cap}
\end{figure}
%cl: network throughput or per cell?

\section{Numerical Results and Discussions}
%In this section, analytical results from the theoretical analysis of transmission probability and network throughput will be illustrated and then compared with simulations.

%Chen:
In this section, we validate the analysis by simulations with Matlab.%end
The BSs and users are randomly deployed over a $1000m \times 1000m$ square with different densities. 
The transmission power of BS is $P_S =1W$, the energy threshold is set to $E_{th}=1*10^{-5}$ or $E_{th}=7*10^{-5}$ joules, the path-loss exponent $\alpha = 3$, and the energy conversion efficiency $\eta = 0.5$.
Besides, the thermal noise is considered negligible compared with the interference from other neighboring BSs.

%In Fig. \ref{fig.lu_trans}, we present the comparison between the analytical and simulation results for the transmission probability under different density of BSs $\lambda_B$.
%Chen
The successful information delivery probability under different density of BSs $\lambda_B$ is shown in Fig.~\ref{fig.lu_trans}. For a smaller $E_{th}=1*10^{-5}$ joule and high density of users, each user may be sufficiently charged in each scheduling round with more users in each cell, and has a successful information delivery probability of 1. The higher the user density, the higher the $P(Tr)$. When the density of users decreases and density of BSs increases, the expected number of users in each cell becomes smaller, as shown in Fig.~\ref{fig.lu_cap_2}. In this case, it may take multiple scheduling rounds for users to harvest energy for one data reception, and thus $P(Tr)$ decreases. For a larger $E_{th}=7*10^{-5}$, $P(Tr)$ increases with $\lambda_B$ as more energy can be harvested from more BSs by a user for data reception, although the number of users in each cell also decreases when $\lambda_B$ increases. For a sufficiently large $\lambda_B$, $P(Tr)$ decreases and eventually converges to $0.5$ when the number of user in each cell approaches 1. In such case, the user takes one slot for charging and one slot for receiving data, and $P(Tr)=0.5$. It is also interesting to see that for the same density of users $\lambda_U=150/km^2,$ as the density of BSs $\lambda_B$ grows, $P(Tr)$ becomes the same after $\lambda_B$ reaches $900/km^2$, which implies that $P(Tr)$ is not dependent on the energy threshold when the density of BSs is large.

The average throughput per cell under different density of BSs $\lambda_B$ is shown in Fig. \ref{fig.lu_cap}. 
BS serves one user at one time in each cell. The achieved per cell throughput is dependent on the achieved link capacity and the successful information delivery probability  $P(Tr)$. %According to~\cite{andrews2011tractable}, as the density of BSs $\lambda_B$ grows, it is known that the SINR of a typical downlink user converge. Thus,
Similar to that in Fig.~\ref{fig.lu_trans}, the average throughput of each cell is mainly dependent on  $P(Tr)$ in a high denscccccccc ity wireless network as SINR converges with a large $\lambda_B$.  % It can be seen that when $E_{th}=7*10^{5}$, the maximum average transmission of BSs will be obtained for an elaborately selected density $\lambda_B$ with given the user density $\lambda_U$. Besides, 
For a given $\lambda_B$, the average throughput increases as the density of users $\lambda_U$ increases.

The total network throughput is plotted in Fig.~\ref{fig.lu_cap_1}. It can be seen that given the base station density $\lambda_B$, the network throughput increases as the density of users $\lambda_U$ increases, which eventually saturates under a sufficiently large $\lambda_U$. In that case, we define the saturation throughput as the sustainable network capacity. In this case, the network operates sustainably, i.e., in each time slot, the user scheduled to receive downlink data has harvested sufficient energy almost surely to attain the maximum network throughput. For a larger $\lambda_B$, more users are required to ensure sufficient energy harvesting in each round to achieve the sustainable network capacity. 

%to be the scenario such that the network throughput is equal to or larger than $99\%$ of the maximum achievable network throughput.
%In that case,
In Fig. \ref{fig.ratio}, we further study the ratio of $\lambda_U/\lambda_B$ when the sustainable network capacity is achieved. It can be observed that the sustainable ratio decreases as the BS density $\lambda_B$ grows. It reveals that  with more BSs, a fewer number of users in each cell can achieve the sustainable network capacity as more energy harvesting from the BSs. 

\begin{figure}[t]
\centering
\includegraphics[width=0.44\textwidth]{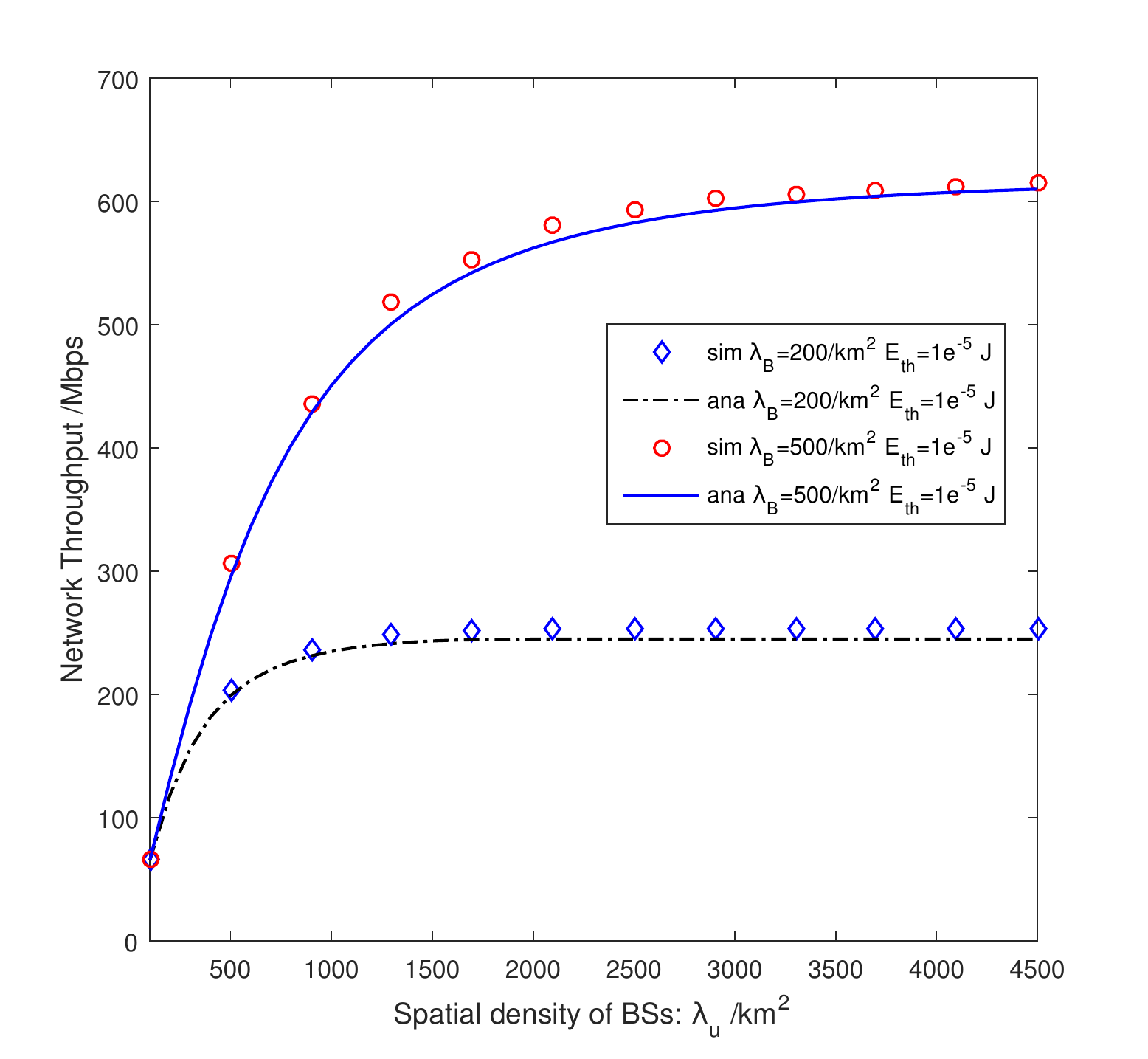}
\caption{Network throughput under different densities of users $\lambda_U$. }
\label{fig.lu_cap_1}
\end{figure}

\begin{figure}[t]
\centering
\includegraphics[width=0.44\textwidth]{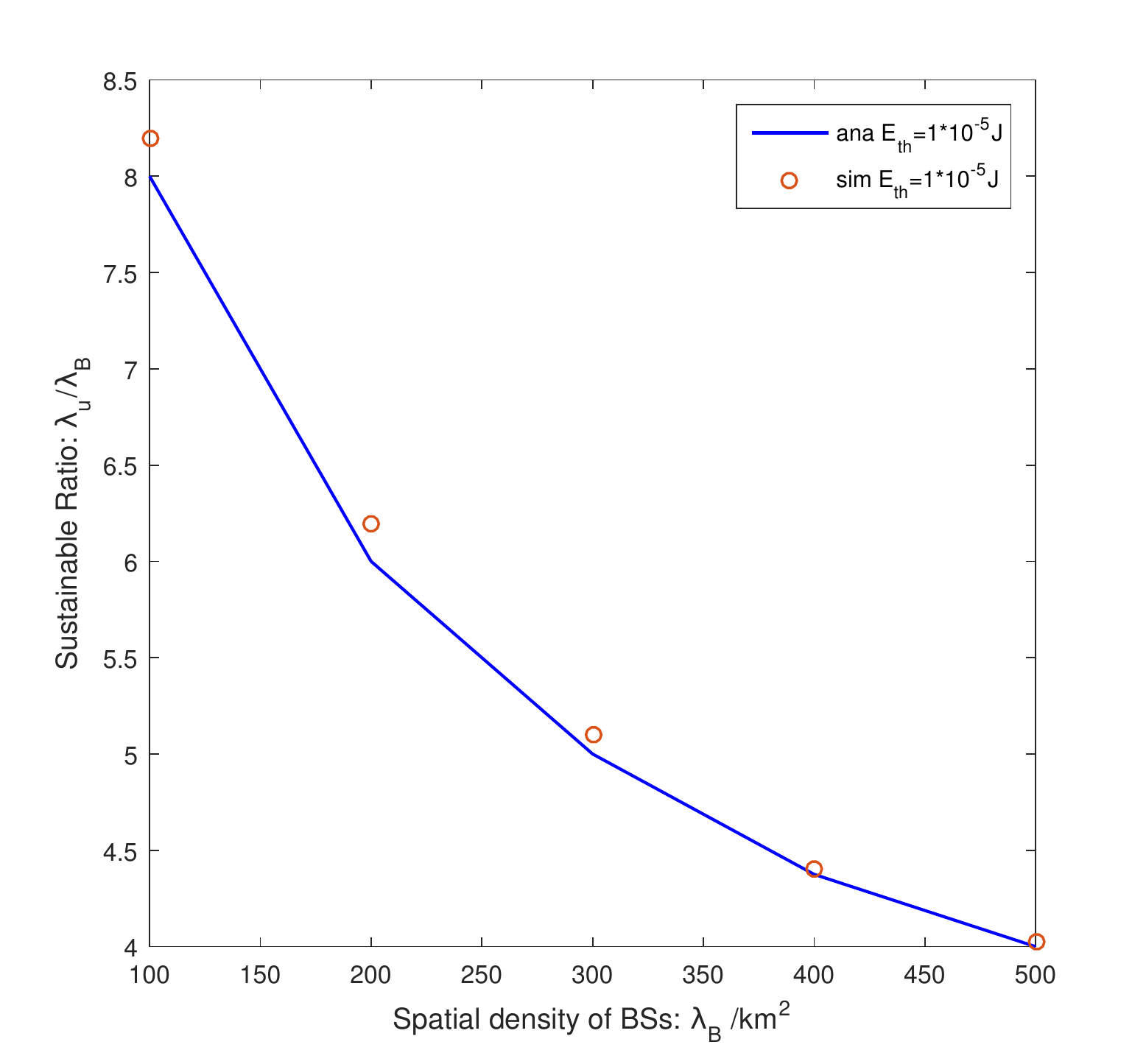}
\caption{Sustainable ratio. }
\label{fig.ratio}
\end{figure}

\section{Conclusion}
In this paper, the network throughput, or the sustainable network capacity of a high density wireless network with RF energy harvesting has been investigated, where low-power wireless devices are powered by the ambient RF energy from concurrent downlink transmissions in different cells. By employing stochastic geometry approach, we have analyzed the successful information delivery probability and the network throughput under different densities of BSs and wireless devices. Our analysis has shown that for a given density of BSs, there exists a minimal density of wireless devices to ensure fully energy sustainable operation of wireless devices; and the maximum sustainable network throughput per cell can be achieved under the optimal density of BSs. Theoretical analysis has been validated by extensive simulations. We will extend the work by considering heterogeneous PPP deployments of BSs and wireless devices in our future work.

\appendix
\subsection{Proof of Lemma \ref{lm.EK}}\label{app.EK}

\begin{figure*}[!t]
\normalsize
\begin{align}
 &\hspace{4pt} \mathbb{P}\left(\sum_{j=1}^{k(n+1)-1} P_{\mathrm{H},j} \geq E_{\mathrm{th}} \bigg | N=n,R=r_1 \right) \nonumber
\hspace{-3pt}\\
=
& \mathbb{P} \left(\sum_{j=1}^{k(n+1)-1} aP_S h_{1,j} r_1^{-\alpha}\!+\hspace{-8pt}\sum_{b_i \in \Phi_{B} /b_{1}} \sum_{j=1}^{k(n+1)-1} aP_S h_{i,j} r_{i}^{-\alpha} \geq E_{\mathrm{th}} \bigg | N=n,R=r_1 \right), \nonumber \\
\hspace{-3pt} \overset{(a)}{=} & \mathbb{P}\left(H_{1,k(n+1)-1} r_1^{-\alpha} + \sum_{b_i \in \Phi_{B} /b_{1}} H_{i,k(n+1)-1} r_{i}^{-\alpha} \geq \frac{E_{\mathrm{th}}}{aP_S} \bigg | N=n,R=r_1 \right), \nonumber \\
\hspace{-3pt} {=} & \mathbb{P}\left(H_{1,k(n+1)-1} r_1^{-\alpha} + \mathbb{E}\left[\sum_{b_i \in \Phi_{B} /b_{1}} H_{i,k(n+1)-1} r_{i}^{-\alpha}\right] \geq \frac{E_{\mathrm{th}}}{aP_S} \bigg | N=n,R=r_1 \right), \nonumber \\
 \hspace{-3pt} \overset{(b)}{\approx}&  \mathbb{P}\left(H_{1,k(n+1)-1}\geq  \frac{E_{\mathrm{th}}r_1^\alpha}{aP_S}-\frac{2\pi [k(n+1)-1]\lambda_Br_1^{2}}{\alpha-2} \bigg| N=n,R=r_1\right), \label{eq.erlang_dist}
\end{align}
\hrulefill
\end{figure*}
According to \eqref{eq.PH}, the probability of energy harvesting ready after $k$ time slots can be written by \eqref{eq.erlang_dist},
where $h_{i,j} \sim \exp(1)$ is the Rayleigh fading from BS $b_i \in \Phi_{B}$ during slot $j$. In step $(a)$, we notice that the summation of independent exponential random variables follows Erlang distribution~\cite{adan2002queueing}, from which we denote $H_{i,k(n+1)-1} = \sum_{j=1}^{k(n+1)-1} h_{i,j} \sim \mathrm{Erlang}(k(n+1)-1,1)$.
In step $(b)$, the total RF energy harvested from beyond the nearest BS $b_1$ can be approximated by its mean with given $r_1$ \cite{kishk2016downlink}, i.e.
\begin{align}
&\mathbb{E}_{\Phi_B}\left[\sum_{b_i \in \Phi_{B}/b_{1}}  H_{i,k(n+1)-1} r_{i}^{-\alpha}  \bigg |R=r_{1}\right] \nonumber\\
 \overset{(c)}{=} &[k(n+1)-1]\mathbb{E}_{\Phi_B}\left[\sum_{b_i \in \Phi_{B}/b_{1}}  r_{i}^{-\alpha} \bigg| R=r_{1}\right], \nonumber \\
 \overset{(d)}{=} &2\pi [k(n+1)-1]\lambda_B\int_{r_1}^{\infty}\frac{1}{r^\alpha}rdr, \nonumber \\
 = &\frac{2\pi [k(n+1)-1]\lambda_B r_1^{2-\alpha}}{\alpha-2}, \nonumber
\end{align}
where $(c)$ holds because all $H_{i,k(n+1)-1}$ are independent Erlang distributed random variable with condition $[k(n+1)-1]$, $(d)$ follows the Campbell's theorem for summation over PPP \cite{haenggi2012stochastic}.

\begin{figure}[t]
\centering
\includegraphics[width=0.44\textwidth]{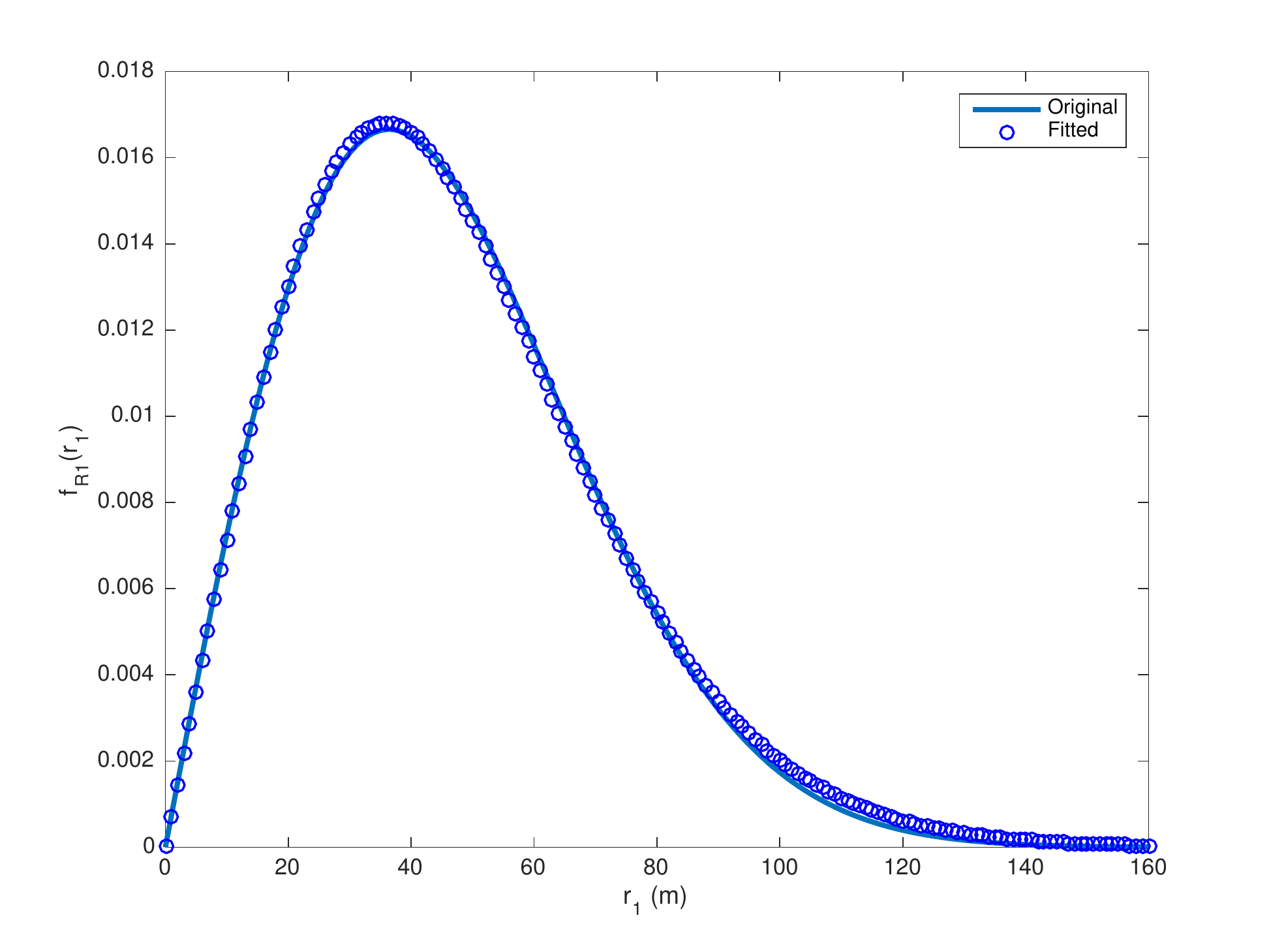}
\caption{Comparison between the original curve and the results after fitting.}
\label{fig.fitting}
\end{figure}

\subsection{Proof of Lemma \ref{lm.X_r1}}\label{app.X_r1}

For the conditional PDF $f_{R_1|X = x}(r_1)$, we know that it will satisfy the following equality
\begin{align} \label{eq.condition}
f_{R_1}(r_1) = \int_{0}^{+\infty} f_{R_1|X=x}(r_1)f_{X}(x)dx,
\end{align}
where $f_X(x) = \frac{3.5^{4.5}}{\Gamma(4.5)}x^{3.5}e^{-3.5x}$ is the PDF of the cell's area where the typical user is located \cite{yu2013downlink}.
Since the PDF $f_X(x)$ comes from the PDF of a Voronoi cell's area that was derived by the Monte Carlo method \cite{ferenc2007size}, it will be very challenging to derive the close-form result for $f_{R_1|X = x}(r_1)$. Thus, noticing that the typical user is uniformly distributed in a cell with area $X=x$, it can be mapped into a normalized cell with unit area, i.e. $X=1$, where the pdf for the distance $r_1$ will be proportional to its area $X=x$. That is to say, we have the following property,
\begin{align} \label{eq.mapping}
f_{R_1|X=x}(r_1) = \frac{f_{R_1|X=1}\left(\frac{r_1}{\sqrt{x}}\right)}{\sqrt{x}}.
\end{align}
Thus, from \eqref{eq.condition} and \eqref{eq.mapping} we can obtain
\begin{align}
f_{R_1}(r_1) = \int_{0}^{+\infty} \frac{f_{R_1|X=1}\left(\frac{r_1}{\sqrt{x}}\right)}{\sqrt{x}}f_{X}(x)dx,
\end{align}
where $f_{R_1}(r_1)$ and $f_{X}(x)$ are both known results and we now can derive the PDF of the distance $r_1$ given the normalized cell area $X=1$.

Since $f_{R_1}(r_1)$ and $f_X(x)$ have the similar form, we can fit the conditional PDF distribution of $r_1$ as follows
\begin{align}
f_{R_1|X=1}(r_1) = c_1 \cdot r_1^{c_2}e^{-c_3 \cdot r_1^{c_4}},
\end{align}
where $c_1 \sim c_4$ are coefficients that we need to calculate by minimizing the fitting error. Considering the least square error fitting, we obtained the fitted curve as shown in Fig. \ref{fig.fitting} and the coefficient values  as $c_1 = 6.029$, $c_2 = 1$, $c_3 = 3.891$ and $c_4 = 2.7$.

\subsection{Proof of Lemma \ref{lm.SINR_threshold}}\label{app.SINR_threshold}

For the typical user, the conditional probability of the link capacity $C$ can be rewritten by
\begin{align}
\mathbb{P} (C \geq t | R=r_1) & = \mathbb{P} (\mathrm{SINR} \geq 2^t - 1 | R=r_1). \label{eq.C_r1}
\end{align}
which is the conditional probability of SINR larger than the threshold $2^t-1$ denoted as $S_{\mathrm{th}}$.
Thus, we can first derive a more general conditional probability for SINR as follows,
\begin{align}
& \hspace{0.5em}\mathbb{P} (\mathrm{SINR} \geq S_{\mathrm{th}} | R=r_1) \nonumber \\
& \hspace{-0.3em} = \mathbb{E}_{h,\Phi_B} \left[\frac{P_S h_1 r_1^{-\alpha}}{\sigma^2+I_{1}} \geq S_{\mathrm{th}} \bigg | R=r_1 \right], \nonumber \\
& \hspace{-0.3em} \overset{(a)}{=} \mathbb{E}_{h,\Phi_B} \left[ e^{-\frac{S_{\mathrm{th}} r_1^{\alpha}(\sigma^2 + I_1)}{P_S} } \bigg | R=r_1 \right],\nonumber \\
% & = e^{-\frac{S_{\mathrm{th}}r_1^{\alpha}\sigma^2}{P_S}} \mathbb{E}_{h,\Phi^{\prime\prime}} \left[  \prod_{b_i \in \Phi_B^{\prime\prime} / b_1} e^{-\frac{S_{\mathrm{th}} r_1^{\alpha}(\sigma^2 + I_1)}{P_S} } \big | r_1 \right], \\
& \hspace{-0.3em} \overset{(b)}{=} e^{-\frac{S_{\mathrm{th}}r_1^{\alpha}\sigma^2}{P_S}} \mathbb{E}_{\Phi_B} \left[ \prod_{b_i \in \Phi_B / b_1} \mathbb{E}_{h}\left[  e^{-S_{\mathrm{th}} r_1^{\alpha} h_i r_i^{-\alpha} } \right] \bigg | R=r_1  \right], \nonumber \\
& \hspace{-0.3em} \overset{(c)}{=} e^{-\frac{S_{\mathrm{th}}r_1^{\alpha}\sigma^2}{P_S}}  \exp\left( -2\pi\lambda_{B}\!\int_{r_1}^{+\infty}\hspace{-4pt}\left( 1\!-\!\mathbb{E}_{h}\left[  e^{- \frac{S_{\mathrm{th}} r_1^{\alpha} h_i}{r^{\alpha} }  } \right]  \right)rdr \right), \nonumber \\
& \hspace{-0.3em} \overset{(d)}{=} e^{-\frac{S_{\mathrm{th}}  r_1^{\alpha}\sigma^2}{P_S} - \pi\lambda_{B}r_1^2 \rho(S_{\mathrm{th}}) },
\end{align}
where $(a)$ holds because $h_1 \sim \exp(1)$, $(b)$ is due the fact that all the fading gains from different BSs are independent and $(c)$ comes from the probability generating functional (PGFL) of the PPP \cite{haenggi2012stochastic}. For step $(d)$, we have
\begin{align}
\int_{r_1}^{+\infty}\left(\!1\!-\!\mathbb{E}_{h}\left[e^{- \frac{S_{\mathrm{th}} r_1^{\alpha} h}{r^{\alpha}} }  \right]\!\right)r dr &\!=\!\int_{r_1}^{+\infty} \left(\!1\!-\!\frac{1}{1+ S_{\mathrm{th}} \frac{r_1^{\alpha}}{r^{\alpha}}}\!\right) rdr,  \nonumber \\
& \!\overset{(e)}{=}\!\frac{1}{2}r_1^2 \rho(S_{\mathrm{th}}), \nonumber
\end{align}
where $(e)$ holds because we change the variables to be $u = (r/r_1)^2S_{\mathrm{th}}^{-\frac{2}{\alpha}}$ and define $\rho(x) = x^{\frac{2}{\alpha}} \int^{+\infty}_{x^{-\frac{2}{\alpha}}} \frac{du}{1+u^{\frac{\alpha}{2}}}$.
Therefore, letting $S_{\mathrm{th}} = 2^t-1$, we can obtain the result in Lemma \ref{lm.SINR_threshold}.

\section*{Acknowledgement}
This work is supported in part by NSF grants CNS-1320736, ECCS-1610874, NSF Career award ECCS1554576, and National Natural Science Foundation of China (NSFC) under grant 61628107.

\vspace{16pt}
\bibliographystyle{IEEEtran}
\bibliography{IEEEfull,geo(1)}

\end{document}